%% file: main.tex
\documentclass[11pt,letterpaper]{article}
\usepackage{amsmath,amssymb,amsfonts,amsthm}
\usepackage{thm-restate}
\usepackage{tikz}
\usepackage{dsfont}
\usepackage{enumitem}
\usepackage{bbm}
\usepackage{xcolor}
\usepackage[ruled,vlined,linesnumbered]{algorithm2e}
\usepackage[margin=1in]{geometry}
\usepackage[colorlinks, citecolor = blue, linkcolor = magenta]{hyperref}
\usepackage[nameinlink]{cleveref}
\usepackage[most]{tcolorbox}   
\tcbuselibrary{skins, breakable}
\usepackage{footnotehyper}

\numberwithin{equation}{section}

\input{macros}

\title{Testing Bipartiteness in Logarithmic Rounds}

\author{Yumou Fei\thanks{Massachusetts Institute of Technology, yumou415@mit.edu. Supported by NSF awards DMS-2022448 and CCF-2310818.} \and
Ronitt Rubinfeld\thanks{Massachusetts Institute of Technology, ronitt@mit.edu. Supported by NSF awards DMS-2022448 and CCF-2310818.}}

\date{}

\begin{document}

\maketitle

\begin{abstract}
The seminal work of Goldreich and Ron (\textit{Combinatorica, 1999}) showed that bipartiteness of bounded-degree graphs can be tested using
$O(\sqrt{n\log n})$ random walks of length $O(\log^{6} n)$. In this work, we improve their result by showing that $O(\sqrt{n})$ random walks of length $O(\log n)$ suffice.
As a corollary, we obtain an $O(\log n)$-pass, $O(\sqrt{n}\log n)$-space streaming algorithm for testing bipartiteness, whose pass complexity is optimal in light of a recent lower bound of Fei, Minzer, and Wang (\textit{arXiv, 2026}).

Our proof takes a different approach from that of Goldreich and Ron, using the semidefinite programming relaxation for Max-Cut introduced by Goemans and Williamson (\textit{J. ACM, 1995}).
\end{abstract}
\tableofcontents

\newpage 

\section{Introduction}

Property testing studies the task of deciding whether a large object satisfies a prescribed property or is far from every object satisfying that property, while inspecting only a small portion of the input. In this paper, we study graph property testing in the \emph{adjacency-list} model, formally defined below.

\begin{definition}
In the adjacency-list model for graph property testing, an algorithm accesses the input graph $G=([n],E)$ through the following two types of oracle queries.
\begin{enumerate}[label=(\arabic*)]
    \item \textit{Degree queries}: given a vertex $v\in[n]$, the oracle returns its degree in $G$, denoted $\deg_G(v)$.
    \item \textit{Local neighbor queries}: given a pair $(v,i)\in[n]\times\bZ$, the oracle returns the $i$th neighbor of $v$ in $G$ if $1\le i\le \deg_G(v)$,\footnote{The ordering of the neighbors of each vertex is arbitrary and may be chosen adversarially.} and returns a special symbol $\perp$ otherwise.
\end{enumerate}
\end{definition}

\begin{definition}\label{def:property_tester}
For $\varepsilon\in(0,1)$, an $\varepsilon$-tester for a graph property $\cP$ (in the adjacency-list model) is a randomized algorithm that, given access to an input graph $G=([n],E)$, achieves the following.
\begin{enumerate}[label=(\arabic*)]
    \item If $G\in\cP$, the tester accepts $G$ with probability at least $2/3$.
    \item If \(
        |E\triangle E'| \ge \varepsilon |E|
    \)
    for every graph $G'=([n],E')\in\cP$ on the same vertex set (in other words, if at least $\varepsilon|E|$ edges must be added or deleted in order to make $G$ satisfy $\cP$), the tester rejects $G$ with probability at least $2/3$. 
\end{enumerate}
\end{definition}

One of the most basic and influential problems in this model is testing bipartiteness. The seminal work of Goldreich and Ron~\cite{goldreich1999sublinear} gave a sublinear-time tester for bipartiteness based on random walks.\footnote{Strictly speaking, the result of~\cite{goldreich1999sublinear} applies only to bounded-degree graphs and regular graphs. The extension to general graphs is due to Kaufman, Krivelevich, and Ron~\cite{kaufman2004tight}.}
Their tester samples a collection of starting vertices, performs random walks from each of them, and rejects if it finds two walks from the same starting vertex that end at the same vertex with different parities. 
Quantitatively, their analysis shows that $O(\sqrt{n\log n})$ random walks of length $O(\log^{6} n)$ suffice. 
This result has become a cornerstone of the area,
inspiring the use of random walks as a fundamental algorithmic primitive in subsequent work on testing graph properties; 
see, for example,~\cite{goldreich2011testing,
BatuFRSW13,
czumaj2010testing,kale2011expansion,nachmias2010testing,czumaj2014finding,czumaj2015testing,kumar2020random,adriaens2023testing}.

Despite the conceptual simplicity of the tester, the analysis in~\cite{goldreich1999sublinear} is quite intricate. When the input graph is a good expander, the argument is simple: Goldreich and Ron showed that $O(\sqrt{n})$ random walks of length $O(\log n)$ already suffice. The difficulty lies in handling graphs with no expansion guarantee. To do so, they prove that the graph can be decomposed into expander-like pieces,\footnote{In the analysis of~\cite{goldreich1999sublinear}, these pieces are not necessarily expanders, but they satisfy enough of the properties needed for the expander-case argument to go through.}
and then argue that only a small fraction of edges go between different pieces. This decomposition-based approach incurs losses of $\operatorname{polylog}(n)$ factors (one source of this loss is that random walks started inside one piece may ``leak'' into other pieces), and makes the proof substantially more complicated than the expander case.

In this work, we give a different analysis of the Goldreich--Ron tester that avoids these losses. In particular, we show that the parameters from the expander case in fact suffice for arbitrary graphs.

\begin{restatable}[Informal]{theorem}{mainquery}\label{thm:main_informal}
In the adjacency-list model, for any $\varepsilon\in (0,1)$ there is an $\varepsilon$-tester for bipartiteness that, given an $n$-vertex simple graph (without isolated vertices), performs
$O_{\varepsilon}(\sqrt{n})$ random walks of length $O_{\varepsilon}(\log n)$.
\end{restatable}

\begin{remark}
The query complexity of the algorithm in~\Cref{thm:main_informal} is $O(\sqrt{n}\log n)$; the best known lower bound on the query complexity is $\Omega(\sqrt{n})$, due to Goldreich and Ron~\cite{goldreich2002property}. When the graph has sufficiently good ``mixing properties,'' Goldreich~\cite{goldreich2017testing} shows that the query complexity can be reduced to $O(\sqrt{n})$.
\end{remark}

To prove~\Cref{thm:main_informal}, we deviate from the original analysis of \cite{goldreich1999sublinear}, and proceed by connecting the Goldreich--Ron tester with the semidefinite programming relaxation of Max-Cut introduced by Goemans and Williamson~\cite{goemans1995improved}, which we will explain in the next subsection. 

\subsection{Discussion: Max-Cut Value}\label{subsec:discussion_max_cut}

A property tester may be viewed as an algorithm that approximates the distance of an input to satisfying a property, albeit in a very coarse sense: it only distinguishes zero distance from sufficiently large distance. In the adjacency-list model, the distance of a graph to bipartiteness is directly related to its Max-Cut value.

\begin{definition}\label{def:max-cut-value}
Given a graph $G=([n],E)$,\footnote{Throughout this paper, graphs have no self-loops but may have parallel edges; equivalently, the edge set is a multiset. A \emph{simple graph} is a graph that has no parallel edges, i.e. every edge has multiplicity 1.}
the Max-Cut value of $G$ is
\begin{equation}\label{eq:def_max_cut}
    \mathrm{val}_{\mathrm{MC}}(G)
    =
    \max_{f:[n]\to\{0,1\}}
    \frac{1}{|E|}
    \sum_{\{u,v\}\in E}\ind{f(u)\neq f(v)}.
\end{equation}
Equivalently, $\mathrm{val}_{\mathrm{MC}}(G)$ is the maximum fraction of edges crossing a bipartition of the vertex set.
\end{definition}

It is easy to see from~\Cref{def:property_tester} that testing bipartiteness with proximity parameter $\varepsilon$ is the same as distinguishing graphs of Max-Cut value $1$ from graphs of Max-Cut value at most $1-\varepsilon$. We formalize this using the following standard promise-problem notation.

\begin{definition}\label{def:max-cut-gap}
Fix a completeness parameter $c\in(0,1]$ and a soundness parameter $s\in[0,c)$.
The promise problem $\mathsf{MaxCut}[c,s]$ is the following decision problem: given a graph $G$, distinguish between (1) the \textit{yes case}, where $\mathrm{val}_{\mathrm{MC}}(G)\ge c$, and (2) the \textit{no case}, where $\mathrm{val}_{\mathrm{MC}}(G)\le s$.
A randomized algorithm must accept with probability at least $2/3$ in the yes case and reject with probability at least $2/3$ in the no case.
\end{definition}

Assuming $\mathbf{P}\neq \mathbf{NP}$ and the Unique Games Conjecture, the range of parameters $c,s$ for which $\mathsf{MaxCut}[c,s]\in \mathbf{P}$ is now well understood~\cite{goemans1995improved,khot2007optimal,o2008optimal}. In particular, given a graph $G=([n],E)$, the semidefinite programming (SDP) rounding algorithm of Goemans and Williamson~\cite{goemans1995improved} outputs (in $\mathrm{poly}(n)$ time) a bipartition of $[n]$ that cuts at least
\(
    0.87856\cdot \mathrm{val}_{\mathrm{MC}}(G)
\)
fraction of the edges of $G$. Consequently, for every $c\in(0,1]$, the promise problem
\(
    \mathsf{MaxCut}[c,\,0.878\,c]
\)
is in $\mathbf{P}$. Moreover, when the completeness parameter $c$ is close to $1$, one can obtain a better soundness guarantee:

\begin{theorem}[\cite{goemans1995improved}]\label{thm:GW_in_P}
For any rational numbers $\varepsilon,\varepsilon'\in(0,1)$ such that  $\varepsilon'>\sqrt{\varepsilon}$, we have
\(
    \mathsf{MaxCut}[1-\varepsilon,1-\varepsilon']\in \mathbf{P}.
\)
\end{theorem}

Goemans and Williamson proved~\Cref{thm:GW_in_P} by showing that if a certain semidefinite programming relaxation of the maximization problem in~\eqref{eq:def_max_cut} has value at least $1-\varepsilon$, then its solution can be rounded to an actual bipartition achieving Max-Cut value at least $1-\sqrt{\varepsilon}$. We formalize this statement in~\Cref{thm:GW}. This theorem will be a crucial tool in the proof of~\Cref{thm:main_informal}, as we now explain.

The original analysis of Goldreich and Ron proceeds by showing that, if the random-walk tester fails to find a violation with constant probability, then the graph must contain a large cut. Without expansion assumptions, however, the most direct consequence of a low rejection probability is only local: for many vertex subsets $V\subseteq[n]$, there exists a good bipartition\footnote{The bipartitions come from distributions of endpoints of even-length and odd-length random walks.} of the subgraph induced by $V$. To patch these local bipartitions into a global cut $f:[n]\to\{0,1\}$, one needs the subsets $V$ to be disjoint. This led Goldreich and Ron to iteratively peel off such subsets after finding local bipartitions. The difficulty is that random walks started in the residual graph may enter previously peeled-off sets $V$. As a result, their analysis requires sufficiently long walks so that the walk returns to the residual graph many times.

Using~\Cref{thm:GW}, we can relax the goal: instead of directly constructing a large cut under the assumption that the random-walk tester rejects with low probability, it suffices to construct a good solution to the Goemans--Williamson SDP. The smoothness of the SDP makes it much easier to combine local information without requiring the underlying vertex subsets to be disjoint. In particular, local SDP solutions can be patched together by using the elementary fact that sums of positive semidefinite matrices remain positive semidefinite.

\begin{remark}\label{rem:not_finding_solutions}
We note that our sublinear-time algorithm itself (which is the same as in~\cite{goldreich1999sublinear}) does not compute an approximate solution to the Goemans--Williamson SDP, nor does it output an approximate Max-Cut. These approximate solutions are only shown to exist in the analysis, under the assumption that the tester rejects with low probability.
\end{remark}

\subsection{Discussion: Multi-Pass Streaming}

In the multi-pass streaming model, the edges of an input graph $G=([n],E)$ are presented as an (adversarially ordered) data stream, and the algorithm is allowed to make multiple passes over the stream while maintaining a small memory space. This sublinear-space model is closely related to the adjacency-list model for sublinear-query algorithms, as evidenced by, for example, a recent work of Fei, Minzer and Wang~\cite{fei2025dichotomy}. In particular,~\Cref{thm:main_informal} essentially implies the following result:

\begin{restatable}{theorem}{mainstreaming}\label{thm:streaming}
For any $\varepsilon\in (0,1)$, there is an $O_{\varepsilon}(\log n)$-pass $O_{\varepsilon}(\sqrt{n}\log n)$-space randomized streaming algorithm for $\mathsf{MaxCut}[1,1-\varepsilon]$ on graphs with $n$ vertices and $\mathrm{poly}(n)$ edges.
\end{restatable}

The $O(\log n)$ pass complexity in~\Cref{thm:streaming} (which improves on the $O(\log ^6 n)$ bound that follows from~\cite{goldreich1999sublinear}) is asymptotically optimal, due to the following lower bound of Fei, Minzer and Wang~\cite{fei2026near}:

\begin{theorem}[\cite{fei2026near}]
Fix an arbitrary constant $\varepsilon\in (0,1)$. Then any $o(\log n)$-pass randomized streaming algorithm for $\mathsf{MaxCut}[1,\frac{1}{2}+\varepsilon]$ on graphs with $n$ vertices and $O(n)$ edges must use at least $n^{1-o(1)}$ bits of memory.
\end{theorem}

\subsection{Related Work and Open Problems}

There has been some interest in obtaining tolerant testers\footnote{As in standard property testing, tolerant testers must reject graphs that are far from satisfying a certain property; but they need to accept any graph close to satisfying the property (instead of only accepting fully satisfying ones).} for bipartiteness in the adjacency-list model. In the terminology used here, tolerant testing can be formulated as the gap problem
\(
    \mathsf{MaxCut}[1-\varepsilon,1-\varepsilon']
\)
for constants $0<\varepsilon<\varepsilon'<1/2$. Existing sublinear-query algorithms for such problems~\cite{peng2023sublinear,jha2024sublinear} require additional assumptions on the input graph, such as expansion or clusterability. It is therefore natural to ask whether the connection between random-walk-based testing and semidefinite programming can be used to design sublinear-time tolerant testers for general graphs.

\begin{problem}
Does there exist a constant $\varepsilon\in(0,1/4)$ such that
\(
    \mathsf{MaxCut}[1-\varepsilon,1/2+\varepsilon]
\)
admits a sublinear-query algorithm in the adjacency-list model?
\end{problem}

We note that query lower bounds of the form $n^{1/2+\Omega_{\varepsilon}(1)}$ are known for
$\mathsf{MaxCut}[1-\varepsilon,1/2+\varepsilon]$, due to~\cite{yoshida2011lower,chiplunkar2018testing}; thus tolerant
testing is strictly harder than standard testing.

Max-Cut is a special case of a constraint satisfaction problem (CSP), and testing bipartiteness in the adjacency-list model can be viewed as a special case of testing satisfiability of CSP instances.\footnote{In the adjacency-list model for CSPs, given a query $(v,i)$, the local neighbor oracle returns the $i$th constraint containing the variable $v$.}
The general question of which CSPs admit sublinear-time satisfiability testers has received some attention~\cite{bogdanov2002lower,yoshida2010query,yoshida2011optimal,yoshida2011lower,aaronson2025property,fei2025unbounded}. In particular, recent work~\cite{fei2025unbounded} shows that testing satisfiability of any \emph{unbounded-width}\footnote{The notion of unbounded width was introduced by Feder and Vardi~\cite{feder1998computational}; we do not define it here.}
CSP requires $\Omega(n)$ queries, leaving open the case of \emph{bounded-width} CSPs. Besides Max-Cut, another prominent bounded-width CSP is $\mathsf{2SAT}$. It would be very interesting to understand whether the satisfiability of $\mathsf{2SAT}$ instances can be tested in sublinear time.

\begin{problem}
Does there exist a constant $\varepsilon\in(0,1/4)$ such that
\(
    \mathsf{Max2SAT}[1,3/4+\varepsilon]
\)
admits a sublinear-query algorithm in the bounded-degree model?\footnote{The bounded-degree model is the adjacency-list model under the additional assumption that every variable participates in at most a constant number of constraints.}
\end{problem}

We remark that for every soundness parameter $s<3/4$, Yoshida~\cite{yoshida2011optimal} gives a constant-query algorithm for $\mathsf{Max2SAT}[1,s]$ in the bounded-degree model. We also note that, due to semidefinite programming, $\mathsf{Max2SAT}$ behaves similarly to $\mathsf{MaxCut}$ in the regime of~\Cref{thm:GW_in_P}; see~\cite{charikar2009near}.

Finally, we remark that the question of bipartiteness testing has been studied in other access models, e.g. in \cite{goldreich1998property,kaufman2004tight}. 
\section{Preliminaries}

In this section, we first recall the Goemans--Williamson SDP~\cite{goemans1995improved} and the Goldreich--Ron algorithm~\cite{goldreich1999sublinear}. In~\Cref{subsec:algorithm}, we state our main technical theorem (\Cref{thm:main}) and show how it implies~\Cref{thm:main_informal,thm:streaming}. Finally, in~\Cref{subsec:notations}, we introduce notation that will be used in the proof of~\Cref{thm:main} in~\Cref{sec:proof_main}.

\paragraph{Notation.}
Throughout the remainder of the paper, let $G=([n],E)$ be a graph with vertex set $[n]$, where $E$ is the multiset of edges of $G$. We assume that $n\ge 2$ and $|E|\ge 1$. Let $A\in\bZ^{n\times n}$ denote the adjacency matrix of $G$: for distinct vertices $i,j\in[n]$, the entry $A_{ij}$ is the multiplicity of the edge $\{i,j\}$ in $E$. Thus $A$ is symmetric, and since $G$ has no self-loops, all diagonal entries of $A$ are zero.

The \emph{simple random walk} on $G$ is the Markov chain with state space $[n]$ which, at each step, chooses an incident edge uniformly at random (counting multiplicities) and moves to the other endpoint of that edge. The \emph{lazy random walk} on $G$ is the Markov chain with state space $[n]$ which, at each step, performs a simple-random-walk step with probability $1/2$ and stays at the current vertex with probability $1/2$. Both the simple random walk and the lazy random walk will be used in our proof, as they provide different conceptual (and analytical) advantages.

\subsection{Semidefinite Programming}

We use the following formulation of the Goemans--Williamson SDP.

\begin{definition}
For a symmetric matrix $A\in\bZ^{n\times n}$, let $\textsc{BasicSDP}_{A}$ denote the following semidefinite program, whose variables are the entries of a positive semidefinite matrix $X\in\bR^{n\times n}$:
\begin{tcolorbox}[
  enhanced,
  breakable,
  title={$\textsc{BasicSDP}_{A}$ for symmetric matrix $A$},
  fonttitle=\bfseries,
  colframe=black,
  colback=white,
  boxrule=0.6pt,
  arc=6pt,
  left=6pt,
  right=6pt,
  top=4pt,
  bottom=4pt
]
\begin{alignat}{3}
\textup{maximize} \quad
    &&  \frac{1}{2}-\frac{\mathrm{tr}(AX)}{2\cdot\mathbf{1}^{T}A\mathbf{1}}
    \label{eq:SDP_objective}\\
\textup{subject to} \quad
    && X_{ii} &\le 1
    \quad && \forall\, i\in[n],
    \label{eq:SDP_no_more_than_1}\\
    && X &\succeq 0. \nonumber
\end{alignat}
\end{tcolorbox}
\end{definition}

\begin{remark}
The quantity $\mathbf{1}^{T}A\mathbf{1}$ appearing in~\eqref{eq:SDP_objective} is the sum of all entries of $A$, and hence equals $2|E|$. In the standard formulation of the Goemans--Williamson SDP, the diagonal constraints are $X_{ii}=1$ rather than $X_{ii}\le 1$. The two formulations are equivalent here: the diagonal entries of $A$ are all zero, so changing the diagonal entries of $X$ does not affect the objective, and increasing diagonal entries preserves positive semidefiniteness.
\end{remark}

We will use the following key property of the Goemans--Williamson SDP.

\begin{theorem}[{\cite[Theorem 3.1.1]{goemans1995improved}}]\label{thm:GW}
For any $\varepsilon\in[0,1]$, if the optimal value of $\textsc{BasicSDP}_{A}$ is at least $1-\varepsilon$, then the Max-Cut value of $G$ is at least $1-\sqrt{\varepsilon}$.
\end{theorem}

\subsection{The Goldreich--Ron Algorithm}\label{subsec:algorithm}

Given the graph $G=([n],E)$, a positive integer $m$ and a parameter $\varepsilon\in(0,1)$, let $\textsc{BipTest}(G,m,\varepsilon)$ denote the algorithm presented in~\Cref{alg:main}. 

\begin{algorithm}
    \DontPrintSemicolon
    \SetKwInOut{Input}{Input}\SetKwInOut{Output}{Output}
    \caption{$\textsc{BipTest}(G,m,\varepsilon)$}\label{alg:main}
    \Input{a graph $G$ with vertex set $[n]$, a positive integer $m$ and a parameter $\varepsilon\in (0,1)$}
    \Output{\texttt{accept} or \texttt{reject}}
    Let $k\gets \left\lceil 10^{3}\varepsilon^{-1}\sqrt{n}\right\rceil$\label{line:alg_k}\;
    Sample $\lceil 4/\varepsilon\rceil$ vertices with probability proportional to their degrees\label{line:sample_edge}\;
    \For{each sampled vertex $v$\label{line:each_vertex_v}}{
    \For{each $i\in\{1,2,\dots,k\}$}{
    Sample an integer $\ell_{i}$ from the binomial distribution $\mathrm{Bin}(m,\frac{1}{2})$\;
    Perform a simple random walk of length $\ell_{i}$ on $G$ starting from $v$\;
    Let $u_{i}$ be the endpoint of the walk\;
    }
    \If{there exist $i,j\in \{1,2,\dots,k\}$ such that $\ell_{i}\not\equiv \ell_{j}\pmod 2$ and $u_{i}=u_{j}$\label{line:odd_cycle}}{
    \Return{\textup{\texttt{reject}} (not bipartite)\label{line:reject}}
    }
    }
    \Return{\textup{\texttt{accept}} (close to bipartite)}
\end{algorithm}

The algorithm $\textsc{BipTest}(G,m,\varepsilon)$ is clearly a \emph{one-sided-error} tester, in the following sense.

\begin{proposition}
If $G$ is bipartite, then for any positive integer $m$ and parameter $\varepsilon\in (0,1)$, the algorithm $\textsc{BipTest}(G,m,\varepsilon)$ accepts with probability 1.
\end{proposition}

\begin{proof}
It is clear that whenever the ``if'' condition in~\Cref{line:odd_cycle} of~\Cref{alg:main} is satisfied, the algorithm has found an odd cycle in $G$, which means $G$ cannot be bipartite.
\end{proof}

We are now ready to state our main technical theorem.

\begin{theorem}[Main theorem]\label{thm:main}
Fix an arbitrary parameter $\varepsilon\in (0,1)$. Let $t=\left\lceil 200\varepsilon^{-8}\log_{2}n\right\rceil$ and $\varepsilon'=\varepsilon^{4}/9$. If $\textsc{BipTest}(G,4t,\varepsilon')$ (\Cref{alg:main}) rejects with probability at most $2/3$, then the Max-Cut value of $G$ is at least $1-\varepsilon$.
\end{theorem}

Since every step of~\Cref{alg:main} can be readily implemented in the multi-pass streaming model, we immediately obtain~\Cref{thm:streaming} (restated below) from~\Cref{thm:main}.

\mainstreaming*
\begin{proof}
Fix a positive integer $t$ and a parameter $\varepsilon\in (0,1)$. Let $\varepsilon'=\varepsilon^{4}/9$ and $k=\left\lceil 10^{3}\varepsilon'^{-1}\sqrt{n}\right\rceil$. Consider the algorithm $\textsc{BipTest}(G,4t,\varepsilon')$. We claim that all $\lceil 4/\varepsilon'\rceil\cdot k$ random walks in $\textsc{BipTest}(G,4t,\varepsilon')$ can be simulated by a streaming algorithm in parallel. Indeed, to advance one step in all the random walks, it clearly suffices to take 2 additional passes over the stream of edges. Therefore the algorithm $\textsc{BipTest}(G,4t,\varepsilon')$ can be implemented in $8t+O(1/\varepsilon')$ passes and $O(\varepsilon'^{-1}k\log n)$ space. 

Plugging~\Cref{thm:main} into the above, we obtain a streaming algorithm for $\mathsf{MaxCut}[1,1-\varepsilon]$ using $O(\varepsilon^{-8}\log n)$ passes and $O(\varepsilon^{-8}\sqrt{n}\log n)$ bits of memory.
\end{proof}

One can similarly deduce~\Cref{thm:main_informal} (restated below) from~\Cref{thm:main}. To implement~\Cref{alg:main} in the adjacency-list model, the only step that requires special attention is~\Cref{line:sample_edge}, where we need to sample vertices with probability proportional to their degrees. 

\mainquery*
\begin{proof}
Fix a parameter $\varepsilon\in (0,1)$. Let 
\[\varepsilon'=\varepsilon^{4}/9,\qquad k=\left\lceil 10^{3}\varepsilon'^{-1}\sqrt{n}\right\rceil,\qquad\text{and}\qquad t=\left\lceil 200\varepsilon^{-8}\log_{2}n\right\rceil.
\]
On any $n$-vertex simple graph $G$ with no isolated vertices, the algorithm of Eden and Rosenbaum~\cite{eden2018sampling} is able to produce an edge of $G$ that is $(\varepsilon'/40)$-close to uniformly random (in total variation distance) over $E$, using $O(\varepsilon'^{-2}\sqrt{n})$ queries in the adjacency-list model. Using the Eden--Rosenbaum algorithm as a subroutine on~\Cref{line:sample_edge} of~\Cref{alg:main}, we can thus approximately simulate $\textsc{BipTest}(G,4t,\varepsilon')$ in the adjacency-list model with additive error at most
\[
\frac{\varepsilon'}{40}\cdot \lceil4/\varepsilon'\rceil\leq \frac{1}{5}.
\]
Furthermore, this approximate simulation of $\textsc{BipTest}(G,4t,\varepsilon')$ still accepts with probability 1 if $G$ is bipartite. Therefore, by~\Cref{thm:main}, this approximate simulation is a one-sided-error $\varepsilon$-tester for the bipartiteness of $G$ with no-case rejection probability at least $\frac{2}{3}-\frac{1}{5}=\frac{7}{15}$. Repeating the simulation twice yields a one-sided-error $\varepsilon$-tester with rejection probability at least \[1-\left(1-\frac{7}{15}\right)^{2}\geq \frac{2}{3}.\]

The final algorithm performs $2k\cdot \lceil 4/\varepsilon'\rceil=O(\varepsilon^{-8}\sqrt{n})$ random walks of length $4t=O(\varepsilon^{-8}\log n)$, and has total query complexity at most $O(\varepsilon'^{-2}\sqrt{n})+8kt\cdot\lceil 4/\varepsilon'\rceil=O(\varepsilon^{-16}\sqrt{n}\log n)$.
\end{proof}

The rest of the paper is devoted to proving~\Cref{thm:main}.

\subsection{Probability Space and Random Walks}\label{subsec:notations}

In this subsection, we set up the notation that will be used in~\Cref{sec:proof_main}.

Recall from~\Cref{line:sample_edge} of~\Cref{alg:main} that we sample vertices of $G$ with probability proportional to their degrees. Throughout the rest of the paper, we denote this distribution by $\mu$.

\begin{definition}\label{def:stationary_distribution}
Let $\mu$ be the probability distribution on $[n]$ defined by
\[
    \mu(\{i\})
    =
    \frac{\deg_G(i)}{2|E|}
    =
    \frac{1}{2|E|}\sum_{j=1}^{n} A_{ij}
    \qquad\text{for all } i\in[n].
\]
\end{definition}

We regard $([n],\mu)$ as a probability space and use the following standard notation.

\begin{definition}\label{def:inner_product_space}
Let $L^{2}(\mu)$ denote the vector space of all functions $[n]\to\bR$, equipped with the inner product
\[
    \inp{f}{g}
    =
    \Exu{i\sim\mu}{f(i)g(i)}
    \qquad\text{for all } f,g:[n]\to\bR.
\]
For $f\in L^{2}(\mu)$ and $p\ge 1$, define
\[
    \|f\|_{p}
    =
    \left(\Exu{i\sim\mu}{|f(i)|^{p}}\right)^{1/p}.
\]
\end{definition}

\begin{definition}
A function $f:[n]\to[0,+\infty)$ satisfying $\|f\|_{1}=1$ is called a \emph{probability density function} with respect to $\mu$. Given such a function $f$, we write $f\d\mu$ for the probability distribution on $[n]$ that assigns probability mass
\(
    f(i)\mu(\{i\})
\)
to each element $i\in[n]$.
\end{definition}

For any vertex $v\in [n]$ that is not isolated in $G$, we have $\deg_{G}(v)\geq 1$ and $\mu(\{v\})>0$. The point-mass distribution at $v$ can then be represented by the following density function.

\begin{definition}\label{def:point_mass}
For any non-isolated vertex $v\in [n]$, we define a probability density function $\delta_{v}:[n]\rightarrow [0,+\infty)$ by
\[
\delta_{v}(u)=\begin{cases}
\mu(\{v\})^{-1},&\text{if }v=u,\\
0,&\text{if }v\neq u.
\end{cases}
\]
Equivalently, let the distribution $\delta_{v}\d\mu$ be the point mass at $v$. 
\end{definition}

The main benefit of the notational framework we have established is that the simple random walk on $G$ can be represented neatly by a linear operator on the inner product space $L^{2}(\mu)$.

\begin{definition}\label{def:Markov_operator}
Define the simple-random-walk operator $P:L^{2}(\mu)\rightarrow L^{2}(\mu)$ as follows: for any function $f\in L^{2}(\mu)$, the function $Pf$ is defined by
\[
Pf(i)=\frac{\sum_{j=1}^{n}A_{ij}f(j)}{\sum_{j=1}^{n}A_{ij}}\qquad\text{for all }i\in [n].
\]
\end{definition}

Note that for any function $f:[n]\rightarrow\bR$, the value of $Pf$ at a vertex $i\in [n]$ is the average $f$-value of the neighbors of $i$. The fact that the simple random walk is a reversible Markov chain will be crucial to our proof. It implies that $P$ is a self-adjoint operator with respect to the inner product on $L^{2}(\mu)$. More concretely, it has the following standard consequence.

\begin{proposition}\label{prop:distribution_of_endpoint}
Fix a non-isolated vertex $v\in [n]$ and a nonnegative integer $r$. If we perform a simple random walk of length $r$ on $G$ starting from $v$, then the distribution of the endpoint of the walk is $(P^{r}\delta_{v})\d\mu$.
\end{proposition}

\begin{proof}
For any vertex $u$, the probability of $u$ under the distribution $(P^{r}\delta_{v})\d\mu$ is
\begin{align*}
\sum_{\substack{(v_{0},v_{1},\dots,v_{r})\in [n]^{r+1}\\ v_{0}=u}}\frac{\prod_{i=0}^{r-1}A_{v_{i}v_{i+1}}}{\prod_{i=0}^{r-1}\deg_{G}(v_{i})}\delta_{v}(v_{r})\mu(\{u\})&=\sum_{\substack{(v_{0},v_{1},\dots,v_{r})\in [n]^{r+1}\\ v_{0}=u,\; v_{r}=v}}\frac{\prod_{i=0}^{r-1}A_{v_{i}v_{i+1}}}{\prod_{i=1}^{r}\deg_{G}(v_{i})}\tag{using~\Cref{def:stationary_distribution,def:point_mass}}\\
&=\sum_{\substack{(v_{r},v_{r-1},\dots,v_{0})\in [n]^{r+1}\\ v_{r}=v,\; v_{0}=u}}\prod_{i=0}^{r-1}\frac{A_{v_{r-i}v_{r-i-1}}}{\deg_{G}(v_{r-i})},\tag{using the fact that $A$ is symmetric}
\end{align*}
which clearly equals the probability that a simple random walk of length $r$ starting from $v$ ends at $u$.
\end{proof}

\section{Proof of Main Result}\label{sec:proof_main}

In this section, we prove the main technical theorem,~\Cref{thm:main}. The proof consists of four main steps, carried out in~\Cref{subsec:birthday,subsec:solution_SDP,subsec:remove_parity,subsec:mixing}, respectively. In~\Cref{subsec:conclude}, the four steps will be combined together to conclude the proof.

\subsection{Birthday Paradox}\label{subsec:birthday}

In each iteration of the ``for'' loop on~\Cref{line:each_vertex_v} of~\Cref{alg:main}, the algorithm collects $O(\sqrt n)$ endpoints of lazy random walks. These endpoints can be partitioned into two sets according to the parity of the number of non-lazy steps taken. If the algorithm does not reject on~\Cref{line:reject}, then these two sets must be disjoint.

This observation naturally leads to a birthday-paradox argument: if we draw $O(\sqrt n)$ samples from each of two distributions on $[n]$ and observe no collision with high probability, then the two distributions must be far apart in total variation distance.

Before carrying out this argument in~\Cref{lem:even_odd_clash}, we first define several operators that will be used there and throughout the rest of the paper.

\begin{definition}\label{def:Q_r_and_R_r}
We define two operators $W,S:L^{2}(\mu)\rightarrow L^{2}(\mu)$ by letting $W=\frac{1}{2}(I+P)$ and $S=\frac{1}{2}(I-P)$. For any nonnegative integer $r$, define two operators $Q_{r},R_{r}:L^{2}(\mu)\rightarrow L^{2}(\mu)$ by
\[
Q_{r}=W^{r}+S^{r}\qquad\text{and}\qquad R_{r}=W^{r}-S^{r}.
\] 
\end{definition}

On a high level, the operator $Q_{r}$ corresponds to even-length random walks, and $R_{r}$ corresponds to odd-length walks. This idea will be made more precise in the proof of~\Cref{lem:even_odd_clash}.

\begin{lemma}\label{lem:even_odd_clash}
Fix a parameter $\varepsilon\in(0,1)$, a non-isolated vertex $v\in[n]$, and a positive integer $m$. Suppose that
\[
    \left\|Q_m\delta_v-R_m\delta_v\big.\right\|_1\leq 2-\varepsilon .
\]
If we perform $k\geq 10^{3}\varepsilon^{-1}\sqrt{n}$ lazy random walks of length $m$ on $G$, each starting from $v$, then with probability at least $2/3$ there exists a vertex $u\in[n]$ that is the endpoint of both a walk with an odd number of non-lazy steps and a walk with an even number of non-lazy steps.
\end{lemma}

\begin{proof}
The number of non-lazy steps in a lazy random walk of length $m$ follows the binomial distribution $\mathrm{Bin}(m,\frac{1}{2})$, and hence is even with probability exactly $\frac{1}{2}$. Furthermore, using~\Cref{prop:distribution_of_endpoint}, it is easy to see that conditioned on the number of non-lazy steps being even, the distribution of the endpoint is given by the probability density function
\[
\left(\left(\frac{I+P}{2}\right)^{m}+\left(\frac{I-P}{2}\right)^{m}\right)\delta_{v}=Q_{m}\delta_{v}.
\]
Similarly, conditioned on the number of non-lazy steps being odd, the distribution of the endpoint is $R_{m}\delta_{v}\d\mu$. Since the total variation distance between the two conditional distribution is
\[
\left\|Q_{m}\delta_{v}\d\mu -R_{m}\delta_{v}\d\mu\big.\right\|_{\mathrm{TV}}=\frac{1}{2}\left\|Q_m\delta_v-R_m\delta_v\big.\right\|_1\leq 1-\frac{\varepsilon}{2},
\]
the desired conclusion follows from standard birthday paradox arguments (\Cref{lem:birthday}).
\end{proof}

Using~\Cref{lem:even_odd_clash}, we can now convert the assumption that~\Cref{alg:main} rejects with low probability into a clean analytic statement:

\begin{corollary}\label{cor:even_odd_clash}
Fix a positive integer $m$ and a parameter $\varepsilon\in (0,1)$. Suppose $\textsc{BipTest}(G,m,\varepsilon)$ (\Cref{alg:main}) rejects with probability at most $2/3$. Then
\[
\Exu{v\sim \mu}{2-\left\|Q_m\delta_v-R_m\delta_v\big.\right\|_1\Big.}\leq 2\varepsilon.
\]
\end{corollary}
\begin{proof}
Assume that the conclusion does not hold. Since $\left\|Q_m\delta_v-R_m\delta_v\big.\right\|_1\leq 2$ for any $v\in [n]$,\footnote{This is because the operator $\frac{1}{2}(Q_{m}-R_{m})=S^{m}$ contracts $L^{1}$-norm. Indeed, since the Markov operator $P$ contracts $L^{1}$-norm, so does $S=\frac{1}{2}(I-P)$.} it follows that
\[
\Pru{v\sim \mu}{2-\left\|Q_m\delta_v-R_m\delta_v\big.\right\|_1\geq \varepsilon\Big.}\geq \frac{\varepsilon}{2}.
\]
Therefore, when $v$ is sampled from $\mu$, with probability at least $\varepsilon/2$ the conclusion of~\Cref{lem:even_odd_clash} holds for $v$. This implies that in each iteration of the ``for'' loop on~\Cref{line:each_vertex_v} of~\Cref{alg:main}, the algorithm rejects with probability at least $\varepsilon/3$ on~\Cref{line:reject}. Since there are $\lceil 4/\varepsilon\rceil$ iterations of this ``for'' loop, and each iteration uses independent randomness, it follows that $\textsc{BipTest}(G,m,\varepsilon)$ accepts with probability at most $(1-\varepsilon/3)^{4/\varepsilon}<1/3$. This contradicts the assumption that $\textsc{BipTest}(G,m,\varepsilon)$ rejects with probability at most $2/3$.
\end{proof}

\subsection{Approximate Solution to the SDP}\label{subsec:solution_SDP}

In this subsection, we construct a feasible solution to the Goemans--Williamson SDP using the random walk operators defined in~\Cref{def:Q_r_and_R_r}. We will then prove in~\Cref{subsec:remove_parity,subsec:mixing,subsec:conclude} that this feasible solution achieves a value close to 1 if~\Cref{alg:main} rejects with low probability.

As argued in the proof of~\Cref{lem:even_odd_clash}, the functions $Q_{m}\delta_{v}$ and $R_{m}\delta_{v}$ are probability density functions for any non-isolated vertex $v$ and positive integer $m$. In particular, they are nonnegative functions, and hence the following definition is justified.  

\begin{definition}\label{def:f_r}
For any non-isolated vertex $v\in [n]$ and any nonnegative integer $r$, we define a function $f_{v}^{(r)}\in L^{2}(\mu)$ by
\[
f^{(r)}_{v}=\sqrt{\frac{1}{2}Q_{r}\delta_{v}}-\sqrt{\frac{1}{2}R_{r}\delta_{v}}.
\]
\end{definition}

Roughly speaking, the function $f_{v}^{(r)}$ corresponds to a ``local bipartition'' of a subset of vertices found in an iteration of Goldreich and Ron's analysis (see~\cite[Section 4.6]{goldreich1999sublinear}). As mentioned in~\Cref{subsec:discussion_max_cut}, the key benefit of working with the SDP relaxation of Max-Cut is that these ``local bipartitions'' can be effortlessly patched together into a positive semidefinite matrix:

\begin{proposition}\label{prop:approximate_solution}
Let $r$ be a fixed positive integer. If we define $X\in \bR^{n\times n}$ by letting\footnote{In \eqref{eq:def_of_X}, the product $f_{v}^{(r)}(i)f_{v}^{(r)}(j)$ can be viewed as the ``opinion'' of the vertex $v$ on whether $i$ and $j$ should be on the same part of a maximum cut; the matrix entry $X_{ij}$ is then the weighted average of the opinions coming from all vertices $v$.}
\begin{equation}\label{eq:def_of_X}
X_{ij}=\Exu{v\sim\mu}{f_{v}^{(r)}(i)f_{v}^{(r)}(j)}
\end{equation}
for each $i,j\in [n]$, then we have 
\begin{enumerate}[label=(\arabic*)]
\item $X$ is positive semi-definite,
\item $X_{ii}\leq 1$ for any $i\in [n]$, and
\item $\mathrm{tr}(AX)=2|E|\cdot \Exu{v\sim \mu}{\inp{f_{v}^{(r)}}{Pf_{v}^{(r)}}}$.
\end{enumerate}
\end{proposition}
\begin{proof}
For any vector $a\in \bR^{n}$, we have
\[
a^{T}Xa=\Exu{v\sim\mu}{\left(\sum_{i=1}^{n}a_{i}f^{(r)}_{v}(i)\right)^{2}}\geq 0.
\]
Combined with the obvious fact that $X$ is symmetric, this implies $X$ is positive semi-definite.

For any $i\in [n]$, we have
\[
X_{ii}=\Exu{v\sim \mu}{f_{v}^{(r)}(i)^{2}}\leq \Exu{v\sim \mu}{\frac{1}{2}Q_{r}\delta_{v}(i)+\frac{1}{2}R_{r}\delta_{v}(i)}= \Exu{v\sim \mu}{W^{r}\delta_{v}(i)\Big.}=W^{r}\mathbf{1}(i)=1.
\]

We also have the direct calculation
\begin{align*}
\mathrm{tr}(AX)&=\sum_{i=1}^{n}\sum_{j=1}^{n}A_{ij}\Exu{v\sim\mu}{f_{v}^{(r)}(i)f_{v}^{(r)}(j)}\tag{using \eqref{eq:def_of_X}}\\
&=\Exu{v\sim\mu}{2|E|\cdot\Exu{i\sim \mu}{f_{v}^{(r)}(i)\cdot \frac{\sum_{j=1}^{n}A_{ij}f_{v}^{(r)}(j)}{\sum_{j=1}^{n}A_{ij}}}}\tag{using~\Cref{def:stationary_distribution}}\\
&=2|E|\cdot \Exu{v\sim\mu}{\inp{f_{v}^{(r)}}{Pf_{v}^{(r)}}},\tag{using~\Cref{def:inner_product_space,def:Markov_operator}}
\end{align*}
which proves the desired identity in the third item.
\end{proof}

\begin{corollary}\label{cor:approximate_solution}
The optimal value of $\textsc{BasicSDP}_{A}$ is at least
\[
1-\frac{1}{2}\cdot\inf_{r\in\bZ,\,r>0}\left\{\Exu{v\sim \mu}{\inp{f_{v}^{(r)}}{Pf_{v}^{(r)}}+1}\right\}.
\]
\end{corollary}

\begin{proof}
For any positive integer $r$, the matrix $X\in \bR^{n\times n}$ defined in \eqref{eq:def_of_X} is a feasible solution to $\textsc{BasicSDP}_{A}$, due to the first and second conclusions of~\Cref{prop:approximate_solution}. The value achieved by this matrix $X$ is
\[
\frac{1}{2}-\frac{1}{4|E|}\cdot\mathrm{tr}(AX)\geq 1-\frac{1}{2}\cdot \Exu{v\sim \mu}{\inp{f_{v}^{(r)}}{Pf_{v}^{(r)}}+1},
\]
where the inequality follows from the third conclusion of~\Cref{prop:approximate_solution}. Therefore, the optimal value of $\textsc{BasicSDP}_{A}$ is at least
\begin{equation}\label{eq:opt_at_least}
1-\frac{1}{2}\cdot \Exu{v\sim \mu}{\inp{f_{v}^{(r)}}{Pf_{v}^{(r)}}+1}
\end{equation}
for any positive integer $r$, as desired.
\end{proof}

Note that for any non-isolated vertex $v\in [n]$ and any positive integer $r$, we have \[\left\|f_{v}^{(r)}\right\|_{2}^{2}\leq \Exu{i\sim \mu}{\frac{1}{2}Q_{r}\delta_{v}(i)+\frac{1}{2}R_{r}\delta_{v}(i)}=\Exu{i\sim \mu}{W^{r}\delta_{v}(i)}=1.\]
Furthermore, the Markov operator $P$ contracts $L^{2}$-norm, and hence
\[
\inp{f_{v}^{(r)}}{Pf_{v}^{(r)}}\geq -\left\|f_{v}^{(r)}\right\|_{2}\cdot \left\|Pf_{v}^{(r)}\right\|_{2}\geq -1.
\]
Therefore, in order for \eqref{eq:opt_at_least} to be close to 1, we need to show that the inner product $\inp{f_{v}^{(r)}}{Pf_{v}^{(r)}}$ is close to $-1$ with high probability over the random vertex $v$ sampled from $\mu$. The next two subsections will be devoted to understanding this inner product. In particular, we will show in~\Cref{subsec:remove_parity} how to relate it to the analytic version of the low-reject-probability assumption derived in~\Cref{cor:even_odd_clash}.

\subsection{Removing Parity}\label{subsec:remove_parity}

In~\Cref{subsec:birthday}, we showed how to use the condition that the number of random walks is $\Omega(\sqrt n)$. To prove~\Cref{thm:main}, we must also exploit the condition that the walks have length $\Omega(\log n)$, which has not yet entered the picture.

Intuitively, once a lazy random walk has run for sufficiently many steps, the distribution of its endpoint becomes approximately stationary, in the sense that taking one additional step changes the distribution only slightly. A natural way to formalize this is to show that, for every non-isolated vertex $v\in[n]$ and every sufficiently large integer $r$, the density functions $W^{r}\delta_v$ and $W^{r+1}\delta_v$ are close in $L^{1}$ distance, or equivalently, that the corresponding endpoint distributions are close in total variation distance. This is the content of the next lemma.

\begin{lemma}\label{lem:Q_vs_PR}
For any positive integer $t$ and any non-isolated vertex $v\in [n]$, we have
\[
\big\|Q_{2t}\delta_{v}-PR_{2t}\delta_{v}\big\|_{1}\leq\sqrt{\frac{2}{t}}\qquad\text{and}\qquad \big\|PQ_{2t}\delta_{v}-R_{2t}\delta_{v}\big\|_{1}\leq\sqrt{\frac{2}{t}}.
\]
\end{lemma}
\begin{proof}
Recall from the proof of~\Cref{lem:even_odd_clash} that $Q_{2t}\delta_{v}\d\mu$ (respectively, $R_{2t}\delta_{v}\d\mu$) is the distribution of the endpoint of the length-$2t$ lazy-random walk starting at $v$ conditioned on the number of non-lazy steps being even (respectively, odd). It follows from~\Cref{lem:binomial_distributions} (and~\Cref{prop:distribution_of_endpoint}) that the total variation distance between $Q_{2t}\delta_{v}\d\mu$ and $PR_{2t}\delta_{v}\d\mu$ is at most $(2t)^{-1/2}$, and hence \(\big\|Q_{2t}\delta_{v}-PR_{2t}\delta_{v}\big\|_{1}\leq \sqrt{2/t}\). Due to the reflection symmetry of the binomial distribution $\mathrm{Bin}(2t,\frac{1}{2})$ about $t$, it also follows from~\Cref{lem:binomial_distributions} that the total variation distance between $PQ_{2t}\delta_{v}\d\mu$ and $R_{2t}\delta_{v}\d\mu$ is at most $(2t)^{-1/2}$, and hence \(\big\|Q_{2t}\delta_{v}-PR_{2t}\delta_{v}\big\|_{1}\leq \sqrt{2/t}\).
\end{proof}

It turns out that~\Cref{lem:Q_vs_PR} does not fully capture the ``mixing'' behavior arising from the $\Omega(\log n)$ walk length; we will return to this point in~\Cref{subsec:mixing}. Nevertheless, the lemma plays an important role in the proof of~\Cref{thm:main}. Its conclusion allows us to relate the inner product \(
    \inp{f_v^{(r)}}{P f_v^{(r)}}
\), which we need to analyze, to the more tractable quantity
\[
    \inp{\sqrt{W^{2t}\delta_v}}{P\sqrt{W^{2t}\delta_v}},
\]
as shown in the next corollary. The latter inner product involves only nonnegative functions, and moreover no longer distinguishes between even- and odd-length random walks.

\begin{corollary}\label{cor:Q_vs_PR}
Fix a positive integer $t$ and a non-isolated vertex $v\in [n]$. If we let $\gamma=2-\big\|Q_{2t}\delta_{v}-R_{2t}\delta_{v}\big\|_{1}$, then we have
\[
\inp{f_{v}^{(2t)}}{Pf_{v}^{(2t)}}+\inp{\sqrt{W^{2t}\delta_{v}}}{P\sqrt{W^{2t}\delta_{v}}}\leq 2\sqrt{\gamma}+2\left(\frac{2}{t}\right)^{1/4}.
\]
\end{corollary}
\begin{proof}
In this proof, we write $g=\frac{1}{2}Q_{2t}\delta_{v}$ and $h=\frac{1}{2}R_{2t}\delta_{v}$. Then by~\Cref{def:Q_r_and_R_r,def:f_r} we have
\begin{equation}\label{eq:root_g-root_h}
\inp{f_{v}^{(2t)}}{Pf_{v}^{(2t)}}=\inp{\sqrt{g}-\sqrt{h}}{P\left(\sqrt{g}-\sqrt{h}\right)}
\end{equation}
and
\begin{equation}\label{eq:root_g+root_h}
\inp{\sqrt{W^{2t}\delta_{v}}}{P\sqrt{W^{2t}\delta_{v}}}=\inp{\sqrt{g+h}}{P\sqrt{g+h}}\leq \inp{\sqrt{g}+\sqrt{h}}{P\left(\sqrt{g}+\sqrt{h}\right)}.
\end{equation}
Adding \eqref{eq:root_g-root_h} and \eqref{eq:root_g+root_h} together yields
\[
\inp{f_{v}^{(2t)}}{Pf_{v}^{(2t)}}+\inp{\sqrt{W^{2t}\delta_{v}}}{P\sqrt{W^{2t}\delta_{v}}}\leq 2\inp{\sqrt{g}}{P\sqrt{g}}+2\inp{\sqrt{h}}{P\sqrt{h}}.
\]
Therefore, it suffices to show that both $\inp{\sqrt{g}}{P\sqrt{g}}$ and $\inp{\sqrt{h}}{P\sqrt{h}}$ are at most $\left(\sqrt{\gamma}+(2/t)^{1/4}\right)/2$. We next prove $\inp{\sqrt{g}}{P\sqrt{g}}\leq \left(\sqrt{\gamma}+(2/t)^{1/4}\right)/2$; the other one can be proved in a similar way.

First note that by~\Cref{def:Markov_operator}, for any $i\in [n]$ we have
\begin{equation}\label{eq:P_rootg_root_Pg}
P\sqrt{g}\,(i)=\frac{\sum_{j=1}^{n}A_{ij}\sqrt{g(j)}}{\sum_{j=1}^{n}A_{ij}}\leq \sqrt{\frac{\sum_{j=1}^{n}A_{ij}\,g(j)}{\sum_{j=1}^{n}A_{ij}}}=\sqrt{Pg}\,(i),
\end{equation}
where we used Cauchy-Schwarz in the second transition. Therefore, we have
\begin{align*}
\inp{\sqrt{g}}{P\sqrt{g}}&\leq \inp{\sqrt{g}}{\sqrt{Pg}}\tag{using \eqref{eq:P_rootg_root_Pg}}\\
&\leq \inp{\sqrt{g}}{\sqrt{h}}+\big\|\sqrt{g}\big\|
_{2}\cdot\left\|\sqrt{h}-\sqrt{Pg}\right\|_{2}\tag{using Cauchy-Schwarz}\\
&\leq \frac{1}{2}\inp{\sqrt{2g}}{\sqrt{2h}}+\frac{1}{\sqrt{2}}\cdot \sqrt{\big\|h-Pg\big\|_{1}}\tag{using the fact that $2g$ is a probability density function and~\Cref{lem:norm_ineq}(1)}\\
&\leq \frac{1}{2}\sqrt{2-\big\|2g-2h\big\|_{1}}+\frac{1}{2}\sqrt{\big\|2h-P(2g)\big\|_{1}}\tag{using~\Cref{lem:norm_ineq}(2)}\\
&\leq \frac{1}{2}\sqrt{\gamma}+\frac{1}{2}\left(\frac{2}{t}\right)^{1/4},\tag{using the definition of $\gamma$ and~\Cref{lem:Q_vs_PR}}
\end{align*}
as desired.
\end{proof}

\subsection{An Interpretation of Mixing}\label{subsec:mixing}

As discussed in~\Cref{subsec:remove_parity}, we need to understand what is special about walks of length $\Omega(\log n)$. The benefit of sufficiently long walks is already partially reflected in~\Cref{lem:Q_vs_PR}, but the applications of that lemma only require the walk length to be at least a constant depending on $\varepsilon$. Intuitively speaking, the correct threshold $\Theta(\log n)$ should correspond to the point at which random walks begin to exhibit a form of ``mixing.''

Since we do not assume that $G$ is an expander, we cannot expect the endpoint distribution of a $\Theta(\log n)$-step random walk to be close to the global stationary distribution $\mu$. Instead, we track the distance of the endpoint distribution from $\mu$. For many natural notions of distance, this distance is nonincreasing with the walk length.\footnote{A non-increasing infinite sequence of nonnegative numbers must converge to a limit. Here, the limit of the distance to $\mu$ (as walk length tends to infinity) may not be zero if the graph $G$ is not connected.} Thus, in this setting, ``mixing'' can be interpreted as the stabilization of this distance rather than convergence to stationarity.

For our purposes, the most suitable notion of distance is relative entropy,\footnote{The original analysis of Goldreich and Ron~\cite{goldreich1999sublinear} uses a similar idea to find a local bipartition of a subset of vertices, but uses the $L^{2}$ distance to capture the relevant mixing behavior.} which is featured in the proof of the following key lemma.

\begin{lemma}\label{lem:JS_Hellinger}
For any non-isolated vertex $v\in [n]$, we have
\[
\sum_{r=0}^{+\infty}\left(1-\inp{\sqrt{W^{r}\delta_{v}}}{P\sqrt{W^{r}\delta_{v}}}\right)\leq \log_{2}\frac{1}{\mu(\{v\})}.
\]
\end{lemma}
\begin{proof}
For any nonnegative function $g:[n]\rightarrow [0,+\infty)$, define its relative entropy to be 
\[\mathrm{Ent}_{\mu}(g):=\Exu{i\sim\mu}{g(i)\log_{2}g(i)\big.}.\]
For any probability density function $g$, we have $\mathrm{Ent}_{\mu}(g)\geq 0$.\footnote{This is due to Jensen's inequality and the convexity of the function $x\mapsto x\log_{2}x$ on $[0,+\infty)$} The goal is to prove that for any probability density function $g:[n]\rightarrow [0,+\infty)$, we have\footnote{Since $\|\sqrt{g}\|_{2}=1$ and $P$ contracts $L^{2}$-norm, the left-hand side of \eqref{eq:JS_Hellinger} is always nonnegative.}
\begin{equation}\label{eq:JS_Hellinger}
1-\inp{\sqrt{g}}{P\sqrt{g}}\leq \mathrm{Ent}_{\mu}(g)-\mathrm{Ent}_{\mu}(Wg).
\end{equation}
Once we have \eqref{eq:JS_Hellinger}, we can replace $g$ with $W^{r}\delta_{v}$ and sum over all nonnegative integers $r$, yielding
\[
\sum_{r=0}^{+\infty}\left(1-\inp{\sqrt{W^{r}\delta_{v}}}{P\sqrt{W^{r}\delta_{v}}}\right)\leq \sum_{r=0}^{+\infty}\Big(\mathrm{Ent}_{\mu}\big(W^{r}\delta_{v}\big)-\mathrm{Ent}_{\mu}\big(W^{r+1}\delta_{v}\big)\Big)\leq \mathrm{Ent}_{\mu}(\delta_{v})=\log_{2}\frac{1}{\mu(\{v\})}.
\]

It remains to prove that \eqref{eq:JS_Hellinger} holds for any probability density function $g$.\footnote{The inequality \eqref{eq:JS_Hellinger} is standard in Markov chain theory, with important consequences such as the implication from log-Sobolev inequalities to entropy contraction. Nevertheless, for the sake of completeness, we choose to provide a full proof of~\eqref{eq:JS_Hellinger}.} We denote $\varphi:[0,+\infty)\rightarrow\bR$ to be the convex function $\varphi(x)=x\log_{2}x$. For any $i\in [n]$, by~\Cref{def:Markov_operator,def:Q_r_and_R_r} we have
\[
\varphi\big(Wg(i)\big)=\varphi\left(\frac{\sum_{j=1}^{n}A_{ij}(g(i)+g(j))/2}{\sum_{j=1}^{n}A_{ij}}\right)\leq \frac{\sum_{j=1}^{n}A_{ij}\,\varphi\big((g(i)+g(j))/2\big)}{\sum_{j=1}^{n}A_{ij}}.
\]
Now we apply~\Cref{lem:two_point_ineq} to the numerator and get
\[
\varphi\big(Wg(i)\big)\leq \frac{1}{\sum_{j=1}^{n}A_{ij}}\cdot \sum_{j=1}^{n}A_{ij}\left(\frac{1}{2}\varphi\big(g(i)\big)+\frac{1}{2}\varphi\big(g(j)\big)+\sqrt{g(i)g(j)}-\frac{1}{2}g(i)-\frac{1}{2}g(j)\right).
\]
Taking expectation over $i\sim \mu$ yields
\[
\big\|\varphi\circ (Wg)\big\|_{1}\leq \frac{1}{2}\big\|\varphi\circ g\big\|_{1}+\frac{1}{2}\big\|P(\varphi\circ g)\big\|_{1}+\inp{\sqrt{g}}{P\sqrt{g}}-\frac{1}{2}\big\|g\big\|_{1}-\frac{1}{2}\big\|Pg\big\|_{1}
\]
Since the operator $P$ preserves $L^{1}$-norm of nonnegative functions, the above can be rewritten as
\[
\mathrm{Ent}_{\mu}(Wg)\leq \mathrm{Ent}_{\mu}(g)+\inp{\sqrt{g}}{P\sqrt{g}}-1,
\]
which is exactly the desired inequality \eqref{eq:JS_Hellinger}.
\end{proof}

The importance of having walks of length $\Omega(\log n)$ finally becomes clear in the following corollary: if $\mu(\{v\})=1/\mathrm{poly}(n)$, we need $t=\Omega(\log n)$ to make the right-hand side of~\eqref{eq:importance_log_n} sufficiently small. 

\begin{corollary}\label{cor:JS_Hellinger}
Fix a positive integer $t$ and a non-isolated vertex $v\in [n]$. If we let $\gamma=2-\big\|Q_{4t}\delta_{v}-R_{4t}\delta_{v}\big\|_{1}$, then we have
\begin{equation}\label{eq:importance_log_n}
\frac{1}{t}\sum_{r=t}^{2t-1}\left(\inp{f_{v}^{(2r)}}{Pf_{v}^{(2r)}}+1\right)\leq 2\sqrt{\gamma}+2\left(\frac{2}{t}\right)^{1/4}+\frac{1}{t}\log_{2}\frac{1}{\mu(\{v\})}.
\end{equation}
\end{corollary}
\begin{proof}
For any $r\in\{t,t+1,\dots,2t-1\}$, due to the fact that the operator $S=\frac{1}{2}(I-P)$ contracts $L^{1}$-norm, we have
\[
\big\|Q_{2r}\delta_{v}-R_{2r}\delta_{v}\big\|_{1}=2\big\|S^{2r}\delta_{v}\big\|_{1}\geq 2\big\|S^{4t}\delta_{v}\big\|_{1}=\big\|Q_{4t}\delta_{v}-R_{4t}\delta_{v}\big\|_{1}=2-\gamma.
\]
It then follows from~\Cref{cor:Q_vs_PR} that
\begin{equation}\label{eq:average_first_part}
\frac{1}{t}\sum_{r=t}^{2t-1}\left(\inp{f_{v}^{(2r)}}{Pf_{v}^{(2r)}}+\inp{\sqrt{W^{2r}\delta_{v}}}{P\sqrt{W^{2r}\delta_{v}}}\right)\leq 2\sqrt{\gamma}+2\left(\frac{2}{t}\right)^{1/4}.
\end{equation}
On the other hand,~\Cref{lem:JS_Hellinger} clearly implies
\begin{equation}\label{eq:average_second_part}
\frac{1}{t}\sum_{r=t}^{2t-1}\left(1-\inp{\sqrt{W^{2r}\delta_{v}}}{P\sqrt{W^{2r}\delta_{v}}}\right)\leq \frac{1}{t}\log_{2}\frac{1}{\mu(\{v\})}.
\end{equation}
Adding \eqref{eq:average_first_part} and \eqref{eq:average_second_part} together yields the conclusion.
\end{proof}

\subsection{Concluding the Proof}\label{subsec:conclude}

The following theorem summarizes what we have proved in~\Cref{subsec:birthday,subsec:solution_SDP,subsec:remove_parity,subsec:mixing}.

\begin{theorem}\label{thm:SDP}
Fix an arbitrary parameter $\varepsilon\in (0,1)$. Let $t$ be an integer such that $t\geq 2\varepsilon^{-2}\log_{2}n$. If $\textsc{BipTest}(G,4t,\varepsilon)$ (\Cref{alg:main}) rejects with probability at most $2/3$, then $\textsc{BasicSDP}_{A}$ has value at least $1-3\sqrt{\varepsilon}$.
\end{theorem}

\begin{proof}
Suppose $\textsc{BipTest}(G,4t,\varepsilon)$ (\Cref{alg:main}) rejects with probability at most $2/3$. By \Cref{cor:even_odd_clash}, it follows that
\[
\Exu{v\sim \mu}{2-\left\|Q_{4t}\delta_v-R_{4t}\delta_v\big.\right\|_1\Big.}\leq 2\varepsilon.
\]
Applying Cauchy-Schwarz yields
\[
\Exu{v\sim \mu}{\sqrt{2-\left\|Q_{4t}\delta_v-R_{4t}\delta_v\big.\right\|_1}}\leq \sqrt{2\varepsilon}.
\]
It then follows from~\Cref{cor:JS_Hellinger} that
\begin{align*}
\frac{1}{t}\sum_{r=t}^{2t-1}\Exu{v\sim \mu}{\inp{f_{v}^{(2r)}}{Pf_{v}^{(2r)}}+1}&\leq 2\sqrt{2\varepsilon}+2\left(\frac{2}{t}\right)^{1/4}+\frac{1}{t}\Exu{v\sim \mu}{\log_{2}\frac{1}{\mu(\{v\})}}\\
&\leq 2\sqrt{2\varepsilon}+2\left(\frac{2}{t}\right)^{1/4}+\frac{1}{t}\log_{2}\left(\Exu{v\sim \mu}{\frac{1}{\mu(\{v\})}}\right)\\
&=2\sqrt{2\varepsilon}+2\left(\frac{2}{t}\right)^{1/4}+\frac{\log_{2}n}{t}\\
&\leq 2\sqrt{2\varepsilon}+2\sqrt{\varepsilon}+\frac{\varepsilon^{2}}{2}\leq 6\sqrt{\varepsilon}.
\end{align*}
Therefore, there exists an $r\in \{t,t+1,\dots,2t-1\}$ such that
\[
\Exu{v\sim \mu}{\inp{f_{v}^{(2r)}}{Pf_{v}^{(2r)}}+1}\leq 6\sqrt{\varepsilon},
\]
and hence the desired conclusion follows from~\Cref{cor:approximate_solution}.
\end{proof}

We are now ready to conclude the proof of~\Cref{thm:main}.

\begin{proof}[Proof of~\Cref{thm:main}]
Apply~\Cref{thm:SDP} to conclude that the value of $\textsc{BasicSDP}_{A}$ is at least $1-\varepsilon^{2}$. It then follows from~\Cref{thm:GW} that the Max-Cut value of $G$ is at least $1-\varepsilon$.  
\end{proof}

\addcontentsline{toc}{section}{References}
\bibliographystyle{alpha}
\bibliography{reference}

\appendix
\section{Some Basic Lemmas}

In this appendix, we collect some basic lemmas used in~\Cref{sec:proof_main} and present their proofs.

\begin{lemma}[Birthday paradox]\label{lem:birthday}
Fix parameters $\varepsilon,\delta\in (0,1)$. Let $p^{(0)},p^{(1)}\in \bR^{n}$ be two vectors that satisfy
\begin{enumerate}[label=(\arabic*)]
\item $p^{(0)}_{i},p^{(1)}_{i}\geq 0$ for any $i\in [n]$,
\item $\sum_{i=1}^{n}p^{(0)}_{i}=\sum_{i=1}^{n}p_{i}^{(1)}=1$, and
\item $\sum_{i=1}^{n}\left|p^{(0)}_{i}-p^{(1)}_{i}\right|\leq 2-\varepsilon$.
\end{enumerate}
Let $\nu$ be the probability distribution on $[n]\times \{0,1\}$ under which the probability of each pair $(i,j)\in [n]\times \{0,1\}$ equals $p^{(j)}_{i}/2$. Then, in \(64\left\lceil
(4\varepsilon)^{-1}\log(2/\delta)\sqrt{n}
\right\rceil\) independent samples from $\nu$, with probability at least $1-\delta$ there exist two sampled pairs $(i,0)$ and $(i,1)$ with the same first coordinate $i\in [n]$ but different second coordinates.
\end{lemma}

\begin{proof}
The given conditions imply
\[
\sum_{i=1}^{n}\min\left\{\frac{1}{2}p^{(0)}_{i},\frac{1}{2}p^{(1)}_{i}\right\}=\sum_{i=1}^{n}\frac{1}{4}\left(p^{(0)}_{i}+p^{(1)}_{i}-\left|p^{(0)}_{i}-p^{(1)}_{i}\right|\right)\geq \frac{\varepsilon}{4}.
\]
The conclusion then follows \cite[Lemma 5.1]{fei2026testing}, or any of the standard birthday-paradox arguments.
\end{proof}

\begin{lemma}[Binomial distributions]\label{lem:binomial_distributions}
Fix a positive integer $t$. Let $X$ be a random variable following the binomial distribution $\mathrm{Bin}(2t,1/2)$. Let $\nu_{0}$ be the distribution of $X$ conditioned on being even, and let $\nu_{1}$ be the distribution of $X+1$ conditioned on $X$ being odd. Then the total variation distance between $\nu_{0}$ and $\nu_{1}$ equals $2^{-2t}\binom{2t}{t}$, which is no more than $(2t)^{-1/2}$.
\end{lemma}

\begin{proof}
For each $i\in \{0,1,\dots,2t\}$, write
\[
    p_i:=\Pr{X=i}=2^{-2t}\binom{2t}{i}
\]
The total variation distance between $\nu_{0}$ and $\nu_{1}$ equals
\begin{align*}
\frac{1}{2}\sum_{i=0}^{t}\big|\nu_{0}(\{2i\})-\nu_{1}(\{2i\})\big|&=\frac{1}{2}\cdot2p_{0}+\frac{1}{2}\sum_{i=1}^{t}|2p_{2i}-2p_{2i-1}|\\
&=p_{0}+\sum_{i=1}^{\lfloor t/2\rfloor}(p_{2i}-p_{2i-1})+\sum_{i=0}^{\lceil t/2\rceil-1}(p_{2t-2i-1}-p_{2t-2i})\\
&=p_{0}+\sum_{i=1}^{\lfloor t/2\rfloor}(p_{2i}-p_{2i-1})+\sum_{i=0}^{\lceil t/2\rceil-1}(p_{2i+1}-p_{2i})\\
&=p_{t}=2^{-2t}\binom{2t}{t}.\qedhere
\end{align*}
\end{proof}

\begin{lemma}[Norm inequalities]\label{lem:norm_ineq}
Let $\mu$ be a probability distribution over $[n]$, and let $f,g\in L^{2}(\mu)$ be two nonnegative functions. Then the following statements hold.
\begin{enumerate}[label=(\arabic*)]
\item We have \(\left\|\sqrt{f}-\sqrt{g}\right\|_{2}\leq \sqrt{\|f-g\|_{1}}\).
\item If $f$ and $g$ are probability density functions, then \( \inp{\sqrt{f}}{\sqrt{g}}\leq \sqrt{2-\|f-g\|_{1}}\),
\end{enumerate}
\end{lemma}

\begin{proof}
The first item follows from pointwise comparison:
\[
\left\|\sqrt{f}-\sqrt{g}\right\|_{2}^{2}=\Exu{i\sim \mu}{\left(\sqrt{f(i)}-\sqrt{g(i)}\right)^{2}}\leq \Exu{i\sim \mu}{\left|f(i)-g(i)\right|\big.}=\|f-g\|_{1}.
\]
For the second item, by Cauchy-Schwarz we have
\begin{align*}
\inp{\sqrt{f}}{\sqrt{g}}^{2}&=\Exu{i\sim \mu}{\sqrt{f(i)g(i)}}^{2}\leq \Exu{i\sim \mu}{\min\{f(i),g(i)\}\big.}\cdot \Exu{i\sim \mu}{\max\{f(i),g(i)\}\big.}\\
&=\Exu{i\sim \mu}{\frac{f(i)+g(i)-|f(i)-g(i)|}{2}\big.}\cdot \Exu{i\sim \mu}{\max\{f(i),g(i)\}\big.}\\
&\leq \left(1-\frac{1}{2}\|f-g\|_{1}\right)\cdot \|f+g\|_{1}\leq 2-\|f-g\|_{1}.\qedhere
\end{align*}
\end{proof}

\begin{lemma}[Information inequality]\label{lem:two_point_ineq}
Let $\varphi:[0,+\infty)\rightarrow \bR$ be the function $\varphi(x)=x\log_{2}x$. Then for any nonnegative numbers $a,b$ we have
\begin{equation}\label{eq:two_point_ineq}
\varphi(a)+\varphi(b)-2\,\varphi\left(\frac{a+b}{2}\right)\geq a+b-2\sqrt{ab}.
\end{equation}
\end{lemma}
\begin{proof}
Without loss of generality assume $a\geq b$. We re-parameterize by letting $a=\lambda (1+x)$ and $b=\lambda(1-x)$, where $\lambda\in [0,+\infty)$ and $x\in [0,1]$. For fixed $x$, both sides of \eqref{eq:two_point_ineq} scales linearly in $\lambda$, so we may assume $\lambda=\frac{1}{2}$. The desired inequality \eqref{eq:two_point_ineq} thus reduces to
\begin{equation}\label{eq:two_point_ineq_reduces}
\varphi\left(\frac{1+x}{2}\right)+\varphi\left(\frac{1-x}{2}\right)+1\geq 1-\sqrt{1-x^{2}},\qquad\text{for all }x\in [0,1].
\end{equation}

We first apply~\Cref{lem:basic_calculus} to get that the function
\[
x\mapsto \frac{\int_{0}^{x}(1-t^{2})^{-1}\d t}{\int_{0}^{x}(1-t^{2})^{-3/2}\d t}=\frac{\frac{1}{2}\ln\frac{1+x}{1-x}}{x(1-x^{2})^{-1/2}}
\]
is nonincreasing in $(0,1)$. Applying~\Cref{lem:basic_calculus} again yields that the function
\[
x\mapsto \frac{\int_{0}^{x}\frac{1}{2}\ln\frac{1+t}{1-t}\d t}{\int_{0}^{x}t(1-t^{2})^{-1/2}\d t}=\frac{\frac{1}{2}(1+x)\ln (1+x)+\frac{1}{2}(1-x)\ln (1-x)}{1-\sqrt{1-x^{2}}}
\]
is nonincreasing in $(0,1)$. This function converges to $\ln 2$ as $x$ goes to 1. Therefore, for any $x\in (0,1)$ we have
\[
\frac{\frac{1}{2}(1+x)\ln (1+x)+\frac{1}{2}(1-x)\ln (1-x)}{1-\sqrt{1-x^{2}}}\geq \ln 2,
\]
which rearranges to
\[
\frac{1+x}{2}\log_{2}\frac{1+x}{2}+\frac{1-x}{2}\log_{2}\frac{1-x}{2}+1\geq 1-\sqrt{1-x^{2}}.
\]
Since both sides of the above inequality are continuous functions on $[0,1]$, this proves \eqref{eq:two_point_ineq_reduces}.
\end{proof}

\begin{lemma}[Basic Calculus]\label{lem:basic_calculus}
Suppose $f_{1},f_{2}:(0,1)\rightarrow(0,+\infty)$ are continuous functions such that $f_{1}(x)/f_{2}(x)$ is nonincreasing in $x$. Then the function $\left(\int_{0}^{x}f_{1}(t)\d t\right)/\left(\int_{0}^{x}f_{2}(t)\d t\right)$ is also nonincreasing in $x$.
\end{lemma}
\begin{proof}
Fix $x,y\in (0,1)$ such that $x<y$, and let
\[
a_{i}=\int_{0}^{x}f_{i}(t)\d t,\qquad b_{i}=\int_{x}^{y}f_{i}(t)\d t,\qquad\text{for each }i\in \{1,2\}.
\]
Since $f_{1}/f_{2}$ is nonincreasing, it is easy to see that $a_{1}/a_{2}\leq f_{1}(x)/f_{2}(x)\leq b_{1}/b_{2}$. Therefore
\[
\frac{a_{1}}{a_{2}}\leq \frac{a_{1}+b_{1}}{a_{2}+b_{2}},
\]
as desired.
\end{proof}
\end{document}

%% file: macros.tex
\newtheorem{theorem}{Theorem}[section]
\newtheorem{lemma}[theorem]{Lemma}
\newtheorem{corollary}[theorem]{Corollary}

\theoremstyle{definition}
\newtheorem{remark}[theorem]{Remark}
\newtheorem{proposition}[theorem]{Proposition}
\newtheorem{definition}[theorem]{Definition}
\newtheorem{problem}[theorem]{Problem}

\renewcommand{\leq}{\leqslant}
\renewcommand{\le}{\leqslant}
\renewcommand{\geq}{\geqslant}
\renewcommand{\ge}{\geqslant}

\newcommand{\cP}{\ensuremath{\mathcal{P}}}

\newcommand{\bE}{\ensuremath{\mathbb{E}}}

\newcommand{\bP}{\ensuremath{\mathbb{P}}}
\newcommand{\bR}{\ensuremath{\mathbb{R}}}

\newcommand{\bZ}{\ensuremath{\mathbb{Z}}}

\newcommand{\Exu}[2]{\underset{#1} \bE \left[ #2 \right] }

\renewcommand{\Pr}[1]{\bP \left[ #1 \right]} 

\newcommand{\Pru}[2]{\underset{ #1 }\bP \left[ #2 \right]}

\renewcommand{\d}{\,\mathrm{d}}
\newcommand{\ind}[1]{\mathds{1} \left[ #1 \right] }

\newcommand{\inp}[2]{\left\langle #1, #2 \right\rangle}

%% file: reference.bib
@article{aaronson2025property,
  title={Property Testing in Bounded Degree Hypergraphs},
  author={Aaronson, Hugo and Carenini, Gaia and Chanda, Atreyi},
  journal={arXiv preprint arXiv:2502.18382},
  year={2025}
}

@inproceedings{adriaens2023testing,
  title={Testing cluster properties of signed graphs},
  author={Adriaens, Florian and Apers, Simon},
  booktitle={Proceedings of the ACM Web Conference 2023},
  pages={49--59},
  year={2023}
}

@article{BatuFRSW13,
  author       = {Tugkan Batu and
                  Lance Fortnow and
                  Ronitt Rubinfeld and
                  Warren D. Smith and
                  Patrick White},
  title        = {Testing Closeness of Discrete Distributions},
  journal      = {J. {ACM}},
  volume       = {60},
  number       = {1},
  pages        = {4:1--4:25},
  year         = {2013},
  url          = {https://doi.org/10.1145/2432622.2432626},
  doi          = {10.1145/2432622.2432626},
  bibsource    = {dblp computer science bibliography, https://dblp.org}
}

@inproceedings{bogdanov2002lower,
  title={A lower bound for testing 3-colorability in bounded-degree graphs},
  author={Bogdanov, Andrej and Obata, Kenji and Trevisan, Luca},
  booktitle={The 43rd Annual IEEE Symposium on Foundations of Computer Science, 2002. Proceedings.},
  pages={93--102},
  year={2002},
  organization={IEEE}
}

@article{charikar2009near,
  title={Near-optimal algorithms for maximum constraint satisfaction problems},
  author={Charikar, Moses and Makarychev, Konstantin and Makarychev, Yury},
  journal={ACM Transactions on Algorithms (TALG)},
  volume={5},
  number={3},
  pages={1--14},
  year={2009},
  publisher={ACM New York, NY, USA}
}

@inproceedings{chiplunkar2018testing,
  title={Testing graph clusterability: Algorithms and lower bounds},
  author={Chiplunkar, Ashish and Kapralov, Michael and Khanna, Sanjeev and Mousavifar, Aida and Peres, Yuval},
  booktitle={2018 IEEE 59th Annual Symposium on Foundations of Computer Science (FOCS)},
  pages={497--508},
  year={2018},
  organization={IEEE}
}

@article{czumaj2010testing,
  title={Testing expansion in bounded-degree graphs},
  author={Czumaj, Artur and Sohler, Christian},
  journal={Combinatorics, Probability and Computing},
  volume={19},
  number={5-6},
  pages={693--709},
  year={2010},
  publisher={Cambridge University Press}
}

@article{czumaj2014finding,
  title={Finding cycles and trees in sublinear time},
  author={Czumaj, Artur and Goldreich, Oded and Ron, Dana and Seshadhri, C and Shapira, Asaf and Sohler, Christian},
  journal={Random Structures \& Algorithms},
  volume={45},
  number={2},
  pages={139--184},
  year={2014},
  publisher={Wiley Online Library}
}

@inproceedings{czumaj2015testing,
  title={Testing cluster structure of graphs},
  author={Czumaj, Artur and Peng, Pan and Sohler, Christian},
  booktitle={Proceedings of the forty-seventh annual ACM symposium on Theory of Computing},
  pages={723--732},
  year={2015}
}

@inproceedings{eden2018sampling,
  title={On Sampling Edges Almost Uniformly},
  author={Eden, Talya and Rosenbaum, Will},
  booktitle={1st Symposium on Simplicity in Algorithms},
  year={2018}
}

@article{feder1998computational,
  title={The computational structure of monotone monadic SNP and constraint satisfaction: A study through Datalog and group theory},
  author={Feder, Tom{\'a}s and Vardi, Moshe Y},
  journal={SIAM Journal on Computing},
  volume={28},
  number={1},
  pages={57--104},
  year={1998},
  publisher={SIAM}
}

@article{fei2025dichotomy,
  title={A dichotomy theorem for multi-pass streaming {CSPs}},
  author={Fei, Yumou and Minzer, Dor and Wang, Shuo},
  journal={arXiv preprint arXiv:2509.11399},
  year={2025}
}

@article{fei2025unbounded,
  title={Unbounded-width {CSPs} are Untestable in a Sublinear Number of Queries},
  author={Fei, Yumou},
  journal={arXiv preprint arXiv:2510.27012},
  year={2025}
}

@article{fei2026near,
  title={Near-Optimal Space Lower Bounds for Streaming {CSPs}},
  author={Fei, Yumou and Minzer, Dor and Wang, Shuo},
  journal={arXiv preprint arXiv:2604.01400},
  year={2026}
}

@article{fei2026testing,
  title={Testing Properties of Edge Distributions},
  author={Fei, Yumou},
  journal={arXiv preprint arXiv:2603.22702},
  year={2026}
}

@article{goemans1995improved,
  title={Improved approximation algorithms for maximum cut and satisfiability problems using semidefinite programming},
  author={Goemans, Michel X and Williamson, David P},
  journal={Journal of the ACM (JACM)},
  volume={42},
  number={6},
  pages={1115--1145},
  year={1995},
  publisher={ACM New York, NY, USA}
}

@article{goldreich1998property,
  title={Property testing and its connection to learning and approximation},
  author={Goldreich, Oded and Goldwasser, Shari and Ron, Dana},
  journal={Journal of the ACM (JACM)},
  volume={45},
  number={4},
  pages={653--750},
  year={1998},
  publisher={ACM New York, NY, USA}
}

@article{goldreich1999sublinear,
  title={A Sublinear Bipartiteness Tester for Bounded Degree Graphs},
  author={Goldreich, Oded and Ron, Dana},
  journal={Combinatorica},
  volume={19},
  number={3},
  pages={335--373},
  year={1999},
  publisher={Springer}
}

@article{goldreich2002property,
  title={Property Testing in Bounded Degree Graphs},
  author={Goldreich, Oded and Ron, Dana},
  journal={Algorithmica},
  volume={32},
  number={2},
  pages={302--343},
  year={2002},
  publisher={Springer}
}

@article{goldreich2011testing,
  title={On Testing Expansion in Bounded-Degree Graphs},
  author={Goldreich, Oded and Ron, Dana},
  journal={Studies in Complexity and Cryptography},
  pages={68},
  year={2011},
  publisher={Springer}
}

@misc{goldreich2017testing,
  title={On Testing Bipartiteness in the Bounded-Degree Graph Model},
  author={Goldreich, Oded},
  year={2017},
  note={\url{https://www.wisdom.weizmann.ac.il/~oded/VO/pt-proj1.pdf}}
}

@inproceedings{jha2024sublinear,
  title={A sublinear time tester for Max-Cut on clusterable graphs},
  author={Jha, Agastya Vibhuti and Kumar, Akash},
  booktitle={51st International Colloquium on Automata, Languages, and Programming (ICALP 2024)},
  pages={91--1},
  year={2024},
  organization={Schloss Dagstuhl--Leibniz-Zentrum f{\"u}r Informatik}
}

@article{kale2011expansion,
  title={An expansion tester for bounded degree graphs},
  author={Kale, Satyen and Seshadhri, Comandur},
  journal={SIAM Journal on Computing},
  volume={40},
  number={3},
  pages={709--720},
  year={2011},
  publisher={SIAM}
}

@article{kaufman2004tight,
  title={Tight bounds for testing bipartiteness in general graphs},
  author={Kaufman, Tali and Krivelevich, Michael and Ron, Dana},
  journal={SIAM Journal on computing},
  volume={33},
  number={6},
  pages={1441--1483},
  year={2004},
  publisher={SIAM}
}

@article{khot2007optimal,
  title={Optimal inapproximability results for {MAX-CUT} and other 2-variable {CSPs}?},
  author={Khot, Subhash and Kindler, Guy and Mossel, Elchanan and O’Donnell, Ryan},
  journal={SIAM Journal on Computing},
  volume={37},
  number={1},
  pages={319--357},
  year={2007},
  publisher={SIAM}
}

@article{kumar2020random,
  title={Random Walks and Forbidden Minors {I}: An $n^{1/2+o(1)}$-Query One-Sided Tester for Minor Closed Properties on Bounded Degree Graphs},
  author={Kumar, Akash and Seshadhri, C and Stolman, Andrew},
  journal={SIAM Journal on Computing},
  volume={52},
  number={6},
  pages={FOCS18--216},
  year={2020},
  publisher={SIAM}
}

@article{nachmias2010testing,
  title={Testing the expansion of a graph},
  author={Nachmias, Asaf and Shapira, Asaf},
  journal={Information and Computation},
  volume={208},
  number={4},
  pages={309--314},
  year={2010},
  publisher={Elsevier}
}

@inproceedings{o2008optimal,
  title={An optimal {SDP} algorithm for Max-Cut, and equally optimal Long Code tests},
  author={{O'Donnell}, Ryan and Wu, Yi},
  booktitle={Proceedings of the fortieth annual ACM symposium on theory of computing},
  pages={335--344},
  year={2008}
}

@inproceedings{peng2023sublinear,
  title={Sublinear-time algorithms for {Max Cut}, {Max E2Lin($q$)}, and {Unique Label Cover} on expanders},
  author={Peng, Pan and Yoshida, Yuichi},
  booktitle={Proceedings of the 2023 Annual ACM-SIAM Symposium on Discrete Algorithms (SODA)},
  pages={4936--4965},
  year={2023},
  organization={SIAM}
}

@article{yoshida2010query,
  title={Query-number preserving reductions and linear lower bounds for testing},
  author={Yoshida, Yuichi and Ito, Hiro},
  journal={IEICE TRANSACTIONS on Information and Systems},
  volume={93},
  number={2},
  pages={233--240},
  year={2010},
  publisher={The Institute of Electronics, Information and Communication Engineers}
}

@inproceedings{yoshida2011lower,
  title={Lower bounds on query complexity for testing bounded-degree {CSPs}},
  author={Yoshida, Yuichi},
  booktitle={2011 IEEE 26th Annual Conference on Computational Complexity},
  pages={34--44},
  year={2011},
  organization={IEEE}
}

@inproceedings{yoshida2011optimal,
  title={Optimal constant-time approximation algorithms and (unconditional) inapproximability results for every bounded-degree {CSP}},
  author={Yoshida, Yuichi},
  booktitle={Proceedings of the forty-third annual ACM symposium on Theory of computing},
  pages={665--674},
  year={2011}
}
